\documentclass[aps,pra,twocolumn,superscriptaddress,amsmath,amssymb,showpacs,notitlepage,longbibliography]{revtex4-1}
\pdfoutput=1
\usepackage{algorithm}
\usepackage{color}
\usepackage{algorithmic}
\usepackage{amsmath}
\usepackage{braket}
\usepackage{multirow}
\usepackage{amssymb}
\usepackage{amsthm}
\usepackage{bm}
\usepackage{thmtools}
\declaretheorem{theorem}
\usepackage{array}
\usepackage{url}
\usepackage{hyperref}
\hypersetup{
    colorlinks=true,
    linkcolor=blue,
    filecolor=magenta,      
    urlcolor=cyan,
    citecolor=red
}
\usepackage{graphicx}
\usepackage{bibentry}
\usepackage{booktabs}
\usepackage{dcolumn}% Align table columns on decimal point
\usepackage{bm}
\usepackage{blkarray}
\usepackage{xparse}
\usepackage{mathtools}
\usepackage{microtype}
\usepackage{enumerate}
\usepackage{soul}

\usepackage{makecell}
\providecommand{\U}[1]{\protect\rule{.1in}{.1in}}
\newtheorem*{theorem*}{Theorem}

\newtheorem{definition}{Definition}

\begin{document}

%\title{On the minimization of the Wigner entropy for Wigner non-negative states}
%\title{A photonic qubit and a subset of Wigner non-negative states verify the minimum Wigner entropy conjecture}
%\title{Subsets of Wigner non-negative states that verify the minimum Wigner entropy conjecture}
%\title{Wigner non-negative states that verify the minimum Wigner entropy conjecture}
\title{Optimal Probe State for Phase Estimation Under Covariant Measurement}

\author{Qipeng Qian}
\affiliation{Program in Applied Mathematics, The University of Arizona, Tucson, Arizona 85721, United States of America}
\author{Christos N. Gagatsos}
\affiliation{Department of Electrical and Computer Engineering, The University of Arizona, Tucson, Arizona, 85721, United States of America}
\affiliation{Wyant College of Optical Sciences, The University of Arizona, Tucson, Arizona, 85721, United States of America}
\affiliation{Program in Applied Mathematics, The University of Arizona, Tucson, Arizona 85721, United States of America}

\begin{abstract}
We study the optimization of input states for phase estimation under covariant measurements. Building on Holevo's framework, which provides the optimal covariant measurement for a fixed input state, we further optimize over the input state itself. For a general even $2\pi$-periodic cost function with non-negative Fourier coefficients, we derive a necessary and sufficient condition for the optimal input state: Its Fock coefficients are determined, up to arbitrary phases, by the eigenvector corresponding to the largest eigenvalue of a Toeplitz matrix defined by the cost function. This characterization yields an explicit expression for the attainable lower bound of the average cost under optimal covariant measurements and shows that this bound asymptotically approaches zero in the infinite-energy limit. For the specific cost function $W(\theta,\tilde{\theta})=4\sin^2[(\theta-\tilde{\theta})/2]$, we obtain the optimal input state and the corresponding minimum average cost in closed form, demonstrating Heisenberg scaling with respect to the mean photon number. 
\end{abstract}

\maketitle

\section{Introduction}\label{sec:intro}

Estimating phase shifts is a fundamental task in quantum information processing since phase plays a critical role in the evolution and manipulation of quantum states. Phase estimation enables key quantum technologies and remains a central paradigm for demonstrating quantum advantage in quantum-enhanced sensing~\cite{Helstrom1976QuantumDetection,Dowling2008QuantumOpticalMetrology,Toth2014QuantumMetrology,Pirandola2018PhotonicQuantumSensing,Polino2020PhotonicQuantumMetrology,Sidhu2020GeometricPerspective,Taylor2016QuantumMetrologyBiology}. Furthermore, phase estimation underpins foundational studies of quantum coherence and interference, making it indispensable for both theoretical and practical advancements in quantum technologies~\cite{Paris2009QuantumEstimation,Giovannetti2011Advances,Pezze2018QuantumMetrology,Streltsov2017Coherence,Caves1981Interferometer}. Over the past several decades, extensive research has focused on probing the ultimate quantum limits for estimating phase shifts and identifying probe states and measurement strategies that can achieve these limits~\cite{Berry2000OptimalStates,Berry2001OptimalInputStates,Monras2006OptimalPhaseMeasurements,Pezze2008MachZehnder,Jarzyna2012QuantumInterferometry,Tsang2012ZivZakai,Rubio2018NonAsymptotic}. 

For the simplest single-mode case, \cite{demkowicz2011optimal} analyzed optimal measurements and probe states within the Bayesian framework, providing an analytical expression for the average cost of estimation with arbitrary prior distributions. By adopting the cost function,
\begin{eqnarray}
W(\theta,\Tilde{\theta})=4\sin^2\frac{\theta-\Tilde{\theta}}{2},
\label{eq:cost-sin}
\end{eqnarray}
they derived optimal measurement strategies that minimize the expected cost. However, their approach involves singular value decomposition of matrices depending on the unknown phase, making both analytical and numerical evaluations challenging, especially as the system dimension increases.

Building upon this framework, \cite{rodríguez2025knowledge} investigated optimal single-mode Gaussian probe states for phase estimation, within Fisher information framework, using homodyne detection. This study highlighted the role of prior uncertainty in determining effective probe states, showing transitions from coherent states to squeezed vacuum states as prior uncertainty decreases and more energy is allocated to squeezing. 

In the present work, we focus on phase estimation of single-mode states while restricting our analysis to covariant measurements. Covariant measurements, which respect the phase-shift symmetry inherent in phase estimation, form a natural and analytically tractable class of measurements, particularly when dealing with symmetric cost functions. Although Holevo's framework \cite{holevo2011probabilistic} provides optimal covariant measurements in this setting, it remains unclear whether explicit expressions for the minimum achievable cost under optimal covariant measurements can be obtained and whether these expressions have any asymptotic behavior. 

In this work, we address this gap by providing an explicit analytical expression for the necessary and sufficient condition for the optimal input state, which leads to the minimum cost under optimal covariant measurements, and by rigorously showing that this minimum cost asymptotically approaches zero in the infinite-energy limit. Furthermore, we demonstrate that the projected coefficients of the optimal input state in this scenario exhibit central symmetry, offering insight on the structure of probe states that approach the ultimate performance limit under covariant measurements. We also compute the optimal probe state under the specific cost function in Eq. (\ref{eq:cost-sin}) and explicitly verify the asymptotic zero-cost behavior, providing practical confirmation of our theoretical findings. 

This paper is organized as follows. Section \ref{sec:Covariant} revisits covariant measurements, explains their relevance to phase estimation, and discusses key theorems from previous works. In Section \ref{sec:lower bound}, we present our analytical result providing an explicit expression for the minimum cost achievable under optimal covariant measurements with appropriately chosen input states and prove that this expression asymptotically approaches zero in the infinite-energy limit (Appendix \ref{appdxa:Central Symmetry} shows the properties of the optimal probe states under covariant measurements). In Section \ref{sec: Single-mode Cases}, we consider the cost function in Eq.~(\ref{eq:cost-sin}) as a specific example and analytically determine the optimal input state and its asymptotic behavior. Finally, Section~\ref{sec:conclusion} summarizes our findings.

\section{Covariant Measurement}\label{sec:Covariant}

The phase estimation problem inherently exhibits phase-shift symmetry, which naturally motivates the consideration of measurement strategies that respect this symmetry. In particular, covariant measurements are well-suited for such scenarios, as they are designed to transform consistently under the same symmetry operations as the parameter being estimated. In this section, we formally introduce covariant measurements and discuss their key properties relevant to phase estimation. This section follows the standard framework of covariant measurements developed by Holevo~\cite{holevo2011probabilistic}. We recall the relevant definitions and results in order to establish the notation and theoretical basis for the analysis in the following sections. 

\begin{definition}[\cite{holevo2011probabilistic}]
    Let $G$ be a parametric group of transformations of a set $\Theta$ and $g \to V_g$
be a (continuous) projective unitary representation of $G$ in a Hilbert space. Let $M(d\theta)$ be a measurement with values in $\Theta$. The measurement $M(d\theta)$ is covariant with respect to representation $g \to V_g$ if 
\begin{eqnarray}
    V_g^{\dag}M(B)V_g=M(B_{g^{-1}}),
\end{eqnarray}
for any $B$ belonging to the $\sigma$-field of Borel subsets of $\Theta$, where 
\begin{eqnarray}
    B_g=\{\theta:\theta=g\theta',\theta'\in B\}.
\end{eqnarray}
\end{definition}

Informally, if the parameter $\theta$ undergoes a transformation $g \in G$, then the outcome statistics of a covariant measurement shift accordingly. However, it is important to note that the measurement outcome $\Tilde{\theta}$ itself is \emph{not} transformed—only the state is. As a result, we compare the outcome distributions in the original coordinate system, leading to a relationship between $p_{g\theta}(\Tilde{\theta})$ and $p_\theta(g^{-1} \Tilde{\theta})$. This is made precise in the following identity,
\begin{eqnarray}
    p_{g\theta} (\Tilde{\theta})
    &=& \mathrm{Tr} [ \rho_{g\theta} M_{\Tilde{\theta}} ] \nonumber\\
    &=& \mathrm{Tr} [ V_{g} \rho_{\theta} V_{g}^{\dagger} M_{\Tilde{\theta}} ] \nonumber\\
    &=& \mathrm{Tr} [ \rho_{\theta} V_{g}^{\dagger} M_{\Tilde{\theta}} V_{g} ] \nonumber\\
    &=& \mathrm{Tr} [ \rho_{\theta} M_{g^{-1}\Tilde{\theta}} ] \nonumber\\
    &=& p_{\theta} (g^{-1}\Tilde{\theta}).
\end{eqnarray}
Thus, for all covariant measurements, the probability distribution of the results $\Tilde{\theta}$ correctly reflects changes in the value of $\theta$ under the group action $g\in G$. This makes it a theoretically valid measurement strategy for parameter estimation.

The same logic directly leads to another important property of covariant measurements: Their performance remains independent of any prior information about the true parameter. The following theorem, as proved in \cite{holevo2011probabilistic}, rigorously establishes this property.
\begin{theorem}[\cite{holevo2011probabilistic}]
\label{thm:1}
    For a covariant measurement $M$ and any cost function $W$ satisfying $W(g\theta,g\Tilde{\theta})=W(\theta,\Tilde{\theta})$, we have 
    \begin{eqnarray}
        \forall \theta,\quad \Bar{C}=C(\theta)=C(0), 
        \label{eq:equiv under prior}
    \end{eqnarray}
    where 
    \begin{eqnarray}
        C(\theta)=\mathrm{Tr}\left[ \int_0^{2\pi}d\Tilde{\theta}M(\Tilde{\theta})W(\theta,\Tilde{\theta})\hat{\rho}_{\theta} \right]
        \label{eq:C(theta)}
    \end{eqnarray}
    is the average cost and
    \begin{eqnarray}
        \Bar{C}=\int d\theta p(\theta)C(\theta)
    \label{eq:ave-cost}
    \end{eqnarray}
    is the Bayesian average cost. 
    \begin{proof}
        \begin{eqnarray}
            C(g\theta)&=&\int d\Tilde{\theta} W(g\theta,g\Tilde{\theta})Tr[\hat{\rho}_{g\theta}M(g\Tilde{\theta})] \nonumber\\
            &=&\int d\Tilde{\theta} W(\theta,\Tilde{\theta})Tr[\hat{\rho}_{\theta}V_g^{\dag}M(g\Tilde{\theta})V_g] \nonumber\\
            &=&\int d\Tilde{\theta} W(\theta,\Tilde{\theta})Tr[\hat{\rho}_{\theta}M(\Tilde{\theta})]\nonumber\\
            &=&C(\theta).
        \end{eqnarray}
        Thus, $C(\theta)$ is independent of $\theta$. 
    \end{proof}
\end{theorem}
We thus see that for covariant measurements, the prior distribution of the parameter of interest is irrelevant to the measurement performance. Building upon this property, Holevo \cite{holevo2011probabilistic} established an explicit construction of optimal covariant measurements for phase estimation, as formalized in the following theorem.
\begin{theorem}[\cite{holevo2011probabilistic}]
\label{thm:2}
The covariant measurement $M(\theta)$ defined by 
\begin{eqnarray}
    \langle n|M(d\theta)|m\rangle=\frac{d\theta}{2\pi}e^{i(m-n)\theta}\frac{\langle n|\psi\rangle\langle \psi|m\rangle}{|\langle n|\psi\rangle||\langle \psi|m\rangle|},
\end{eqnarray}
is the optimal covariant measurement of angle of rotation $\theta$ for states that can be written as 
\begin{eqnarray}
    e^{-i\theta \hat{N}}|\psi\rangle\langle\psi|e^{i\theta \hat{N}},
\end{eqnarray}
where $\hat{N}$ is the photon number operator, $\{|m\rangle\}$ are the eigenstates of $\hat{N}$, and for any even $2\pi$-periodic cost function $W$ that can be represented by a Fourier series and satisfies 
\begin{eqnarray}
    W(\theta,\Tilde{\theta})=W(\theta-\Tilde{\theta})&=&w_0-\sum_{k=1}^{\infty}w_k\cos\left(k(\theta-\Tilde{\theta})\right);\nonumber\\ 
    w_k&\geq&0,\quad \forall k\geq1.
    \label{eq: W}
\end{eqnarray}
\end{theorem}

Combining Theorems \ref{thm:1} and \ref{thm:2}, we have the average cost $\Bar{C}$ given by the optimal measurement as 
\begin{eqnarray}
    \Bar{C}=w_0-\sum_{k=1}^{\infty}w_k\Big( \sum_{n=0}^{\infty}|\langle\psi|n\rangle||\langle n+k|\psi\rangle| \Big). 
    \label{eq:cost of optimal M}
\end{eqnarray}

\section{Attainable lower bound for optimal covariant measurement}
\label{sec:lower bound}

Holevo's result reviewed in the previous section determines the optimal covariant measurement for a fixed input state. In this section, we take the next step and optimize the input state itself within this covariant-measurement framework. Specifically, starting from the average cost in Eq.~(\ref{eq:cost of optimal M}), we derive a necessary and sufficient condition on the photon-number amplitudes of the input state that minimizes the cost under the optimal covariant measurement. This also leads to an explicit expression for the attainable lower bound of the average cost. 

Denoting $c_n=\langle\psi|n\rangle$, where $|n\rangle$ is a Fock vector, we have
\begin{eqnarray}
\Bar{C}^*
&=&\min_{|\psi\rangle}\Bar{C} \nonumber\\
&=&w_0-\max \sum_{k=1}^{N} w_k \left[\sum_{m=0}^{N-k} |c_m| |c_{m+k}|\right] \nonumber\\
&=&w_0-\max \sum_{i,j=0, i<j}^{N} w_{j-i} |c_i| |c_{j}| \nonumber\\
&=&w_0-\frac{1}{2} \max \sum_{i,j=0}^{N} w_{|i-j|} |c_i| |c_{j}| \nonumber\\
&=&w_0-\frac{1}{2} \max |c|^T\mathbf{W}|c|, 
\label{eq:cost_matrix_form}
\end{eqnarray}
where $|c|=(|c_0|,\ldots,|c_N|)^T$ and $\mathbf{W}_{ij}=w_{|i-j|}$ if $i\neq j$, while $\mathbf{W}_{ij}=0$ if $i=j$. The normalization of the input state implies $\| |c| \|_2=1$. Therefore, the optimization of the input state has been reduced to maximizing a quadratic form over normalized non-negative vectors. This observation leads to the following characterization of the optimal input state and the attainable lower bound. 

\begin{theorem}[Attainable lower bound under optimal covariant measurements]
\label{thm:attainable_lower_bound}
For input states supported on $\mathrm{span}\{|0\rangle,\ldots,|N\rangle\}$, let $\mathbf{W}$ be the Toeplitz matrix defined in Eq.~\eqref{eq:cost_matrix_form}. Then the minimum average cost achievable by optimizing the input state under the optimal covariant measurement is
\begin{eqnarray}
    \Bar{C}^*=w_0-\frac{\lambda_{\max}}{2}, 
    \label{eq:barC_star_lambda}
\end{eqnarray}
where $\lambda_{\max}$ is the largest eigenvalue of $\mathbf{W}$. Moreover, an optimal input state is obtained by choosing
\begin{eqnarray}
    |c_n|=\nu^{\max}_n,\qquad n=0,\ldots,N, 
    \label{eq:optimal_amplitudes_general}
\end{eqnarray}
where $\nu^{\max}$ is a normalized eigenvector associated with $\lambda_{\max}$ that can be chosen to be entrywise non-negative, i.e., $\nu^{\max}_n\geq 0$ for all $n$. The phases of the coefficients $c_n$ are arbitrary. 
\end{theorem}

\begin{proof}
Since $\mathbf{W}$ is real and symmetric, the Rayleigh--Ritz variational principle gives 
\begin{eqnarray}
    \max_{\|x\|_2=1} x^T\mathbf{W}x=\lambda_{\max}. 
\end{eqnarray}
In Eq.~\eqref{eq:cost_matrix_form}, however, the vector to be optimized is $|c|$, whose entries are non-negative. Since $\mathbf{W}$ is entrywise non-negative, the Perron--Frobenius theorem~\cite{Perron1907ZurTD,frobenius1912matrizen} ensures that an eigenvector associated with $\lambda_{\max}$ can be chosen to be entrywise non-negative. Hence this eigenvector is admissible as $|c|$. Choosing $|c|=\nu^{\max}$ therefore achieves the maximum of the quadratic form in Eq.~\eqref{eq:cost_matrix_form}, giving 
\begin{eqnarray}
    \Bar{C}^*=w_0-\frac{\lambda_{\max}}{2}.
\end{eqnarray}
This proves the claim. 
\end{proof}

Since the optimal input state is determined by the Perron eigenvector $\nu^{\max}$ of the Toeplitz matrix $\mathbf{W}$, structural properties of $\nu^{\max}$ directly translate into structural properties of the optimal photon-number amplitudes. In Appendix~\ref{appdxa:Central Symmetry}, we show that for the finite-dimensional Toeplitz matrices considered here, this eigenvector is centrally symmetric; equivalently, with the indexing convention of Eq.~\eqref{eq:cost_matrix_form}, the optimal amplitudes satisfy 
\begin{eqnarray}
    |c_n|=|c_{N-n}|. 
\end{eqnarray}

Moreover, notice that when $N\to\infty$, $\mathbf{W}$ is an infinite-dimensional Toeplitz matrix and $\{w_k\}_{k=1}^{\infty}$ is absolutely summable as long as $W(0)<\infty$. By Theorem 1.1 of~\cite{bottcher2000toeplitz}, we have 
\begin{eqnarray}
    \lambda_{\max}
    &=&\mathbf{L}^{\infty}\left[\sum_{k=-\infty}^{\infty}w_{|k|}e^{ik\theta}\right] \nonumber\\
    &=&\mathbf{L}^{\infty}\left[2\sum_{k=1}^{\infty}w_{k}\cos(k\theta)\right] \nonumber\\
    &=&2\sum_{k=1}^{\infty} w_k, 
    \label{eq:lambda_max}
\end{eqnarray}
where $\mathbf{L}^\infty[f] := \sup_x |f(x)|$. 

Returning to the Fourier representation of the cost function $W$, we can see that as long as $W(0)=0$, we obtain
\begin{eqnarray}
    \Bar{C}^*=w_0-\sum_{k=1}^{\infty}w_k=0, 
\end{eqnarray}
which indicates that, for the sequence of input states constructed from $\nu^{\max}$, the optimal covariant measurement asymptotically yields the best possible estimation performance.

\section{Examples and the best input state}\label{sec: Single-mode Cases}

We now illustrate the general result of Section~\ref{sec:lower bound} with the cost function in Eq. (\ref{eq:cost-sin}). For this cost function, the Fourier coefficients take the simple form $w_0=2$, $w_1=2$, and $w_k=0$ for all $k\geq 2$. Consequently, the Toeplitz matrix $W$ becomes tridiagonal, which allows us to obtain both the optimal input state and the corresponding minimum average cost in closed form. This example also makes explicit the asymptotic scaling of the optimal covariant strategy in terms of the mean photon number. 

For this choice of coefficients, the finite-dimensional Toeplitz matrix has nonzero entries only on the first off-diagonals. Assuming an eigenvector with components $v_j=\sin(j\theta)$ and boundary conditions $v_{0}=v_{N+1}=0$, we have 
\begin{eqnarray}
    (\mathbf{W}\nu)_j&=&2v_{j-1}+2v_{j+1} \nonumber\\
    &=&4\sin(j\theta)\cos(\theta) \nonumber\\
    &=&\lambda v_{j},
\end{eqnarray}
where $\lambda=4\cos(\theta)$. Taking into account $v_{0}=v_{N+1}=0$, we can set $\theta_n=\frac{n\pi}{N+1}$ to obtain $N$ different eigenvalues of $W$. We then have the largest positive eigenvalue 
\begin{eqnarray}
    \lambda_{max}=4\cos(\frac{\pi}{N+1})
\end{eqnarray}
and the corresponding positive eigenvector 
\begin{eqnarray}
    \nu^{max}_j=\frac{1}{A}\sin(\frac{j\pi}{N+1}),\quad j=1,2,...,N,
    \label{eq:nu_{max}}
\end{eqnarray}
where $A=\sqrt{\sum_{j=1}^N\sin^2(\frac{j\pi}{N+1})}$. 

Thus, for any input state $|\psi\rangle$ satisfying 
\begin{eqnarray}
    |\langle\psi|n\rangle|=\frac{1}{A}\sin(\frac{(n+1)\pi}{N+1}),\quad n\in[0,N-1],
\end{eqnarray}
we have 
\begin{eqnarray}
    \Bar{C}^*=w_0-\frac{\lambda_{max}}{2}=2-2\cos(\frac{\pi}{N+1}),
\end{eqnarray}
which tends to $0$ as $N\to\infty$. This is consistent with the results in the previous section. 

Furthermore, when $N\to\infty$, we have 
\begin{eqnarray}
    \Bar{n}&=&\lim_{N\to\infty}\frac{\sum_{n=0}^{N-1}\sin^2(\frac{(n+1)\pi}{N+1})n}{A^2}\nonumber\\   
    &=&\frac{N-1}{2}.
\end{eqnarray}
Also, since $\cos(x)\approx1-\frac{x^2}{2}$ when $x\to0$, we have 
\begin{eqnarray}
    \Bar{C}^*&=&2-2\cos(\frac{\pi}{2(\Bar{n}+1)})\label{eq:c of energy}\\
    &\approx&2-2(1-\frac{(\frac{\pi}{2(\Bar{n}+1)})^2}{2})\nonumber\\
    &=&\frac{\pi^2}{4(\Bar{n}+1)^2}.
    \label{eq:decrease rate}
\end{eqnarray}
Thus, the estimation error asymptotically scales as $\Bar{C}^*\sim\frac{1}{\Bar{n}^2}$. Since $\Bar{C}^*\to0$ when $N\to\infty$, by Eq. (\ref{eq:C(theta)}), we know that $W(\theta,\hat{\theta})=4\sin^2(\frac{\theta-\hat{\theta}}{2})\to0$, which implies $\theta-\hat{\theta}\to0$. For $\theta-\hat{\theta}\to0$, we have 
\begin{eqnarray}
    W(\theta,\hat{\theta})\approx(\theta-\hat{\theta})^2.
    \label{eq:W approx}
\end{eqnarray}
Thus, Eq. (\ref{eq:decrease rate}) can be approximately interpreted as $\mathbb{E}[(\theta-\hat{\theta})^2]$, which then demonstrates the Heisenberg scaling.

\section{Conclusion}\label{sec:conclusion}

In this paper, we derived the necessary and sufficient condition for the optimal input state under optimal covariant measurements in Theorem \ref{thm:attainable_lower_bound}. Specifically, the photon-number amplitudes of the optimal input state are determined, up to arbitrary phases, by the eigenvector corresponding to the largest eigenvalue of the Toeplitz matrix defined by the Fourier coefficients of the cost function. This characterization also yields an explicit form of the attainable lower bound for the average cost, which asymptotically approaches zero as the input energy increases. We also prove a central symmetry property for the optimal input state under general cost functions, providing structural insight into state design for covariant strategies. As a concrete example, we consider the specific cost function of Eq. (\ref{eq:cost-sin}), for which we explicitly construct the optimal input state and calculate the corresponding minimum average cost, thereby demonstrating the Heisenberg scaling for this case.

We anticipate that future works will elaborate on the level of practicality of covariant measurements and compare their performance to well-established measurements like homodyne detection. Also, we leave for future work the generalization to multiple modes.

\newpage
\onecolumngrid
\appendix
\section{Central Symmetry of the Eigenvector corresponding to the Largest Eigenvalue for Toeplitz Matrices}\label{appdxa:Central Symmetry}

Although we have shown in the previous section that the optimal input states are determined by $\nu^{max}$, its analytical form is generally inaccessible. In this subsection, we aim to characterize key properties of $\nu^{max}$, thereby gaining insights into the structure of the optimal input states, providing guidance for state design in covariant strategies.

Consider any finite-dimensional Toeplitz matrix satisfying
\begin{eqnarray}
    \mathbf{W}_{ij} = w_{|i-j|}, \quad w_0 = 0. 
\end{eqnarray}
We aim to prove that the eigenvector corresponding to the maximum eigenvalue of $\mathbf{W}$ is symmetric about the center, i.e.,
\begin{eqnarray}
    v_k = v_{N - k + 1}. 
\end{eqnarray}

Define the reflection matrix $\mathbf{P}$ by:
\begin{eqnarray}
    \mathbf{P}_{ij} = \delta_{i, N - j + 1}, 
\end{eqnarray}
which reverses the order of entries in a vector:
\[
\mathbf{P} v = (v_N, v_{N-1}, \dots, v_1)^T.
\]

For the Toeplitz matrix, we have:
\begin{eqnarray}
    (\mathbf{P} \mathbf{W} \mathbf{P})_{ij} = \mathbf{W}_{N - i + 1, N - j + 1}. 
\end{eqnarray}
Since
\begin{eqnarray}
    | (N - i + 1) - (N - j + 1) | = | j - i | = | i - j |, 
\end{eqnarray}
we have 
\begin{eqnarray}
    \mathbf{W}_{N - i + 1, N - j + 1} = w_{|i - j|} = \mathbf{W}_{ij}. 
\end{eqnarray}
Since $\mathbf{P}^2 = I$, we obtain:
\begin{eqnarray}
    \mathbf{P} \mathbf{W} \mathbf{P} = \mathbf{W} \quad \implies \quad [\mathbf{W}, \mathbf{P}] = 0. 
\end{eqnarray}

Let $\lambda_{\max}\in\mathbb{R}$ denote the unique maximum eigenvalue of $W$, with corresponding eigenvector $\nu^{max}$:
\begin{eqnarray}
    \mathbf{W} \nu^{max} = \lambda_{\max} \nu^{max}. 
\end{eqnarray}
Since $\mathbf{W}$ and $\mathbf{P}$ commute, we have:
\begin{eqnarray}
    \mathbf{W} (\mathbf{P} \nu^{max}) = \mathbf{P} \mathbf{W} \nu^{max} = \lambda_{\max} (\mathbf{P} \nu^{max}),
\end{eqnarray}
which indicates that $\mathbf{P} \nu^{max}$ is also an eigenvector associated with $\lambda_{\max}$.

Now, since $\mathbf{W}$ is irreducible and non-negative by construction, using the Perron–Frobenius theorem for irreducible non-negative matrices, $\lambda_{max}$ is unique and simple. Thus, $\nu^{max}$ is unique up to scaling and can additionally be chosen to have positive entries. Hence, since $\mathbf{P} \nu^{max}$ and $\nu^{max}$ are both eigenvectors for $\lambda_{\max}$ with positive entries, there exists a scalar $c$ such that:
\begin{eqnarray}
    \mathbf{P} \nu^{max} = c \nu^{max}.
\end{eqnarray}
Again, using the fact $\mathbf{P}^2=I$, we have 
\begin{eqnarray}
    |c| = 1. 
\end{eqnarray}
Since both $\mathbf{P} \nu^{max}$ and $\nu^{max}$ have positive entries, we can conclude that 
\begin{eqnarray}
    c = 1. 
\end{eqnarray}
Thus, we have 
\begin{eqnarray}
    \mathbf{P} \nu^{max} = \nu^{max}, 
\end{eqnarray}
which implies
\begin{eqnarray}
    \nu^{max}_k = \nu^{max}_{N - k + 1}.
\end{eqnarray}
This establishes the central symmetry of the maximum eigenvector under our construction of the Toeplitz matrices. Consequently, for the optimal input state $|\psi\rangle$ with $c_n=\langle\psi|n\rangle$, we should expect 
\begin{eqnarray}
    |c_{k}|=|c_{N - k + 1}|. 
\end{eqnarray}

\twocolumngrid
\bibliography{refs.bib}

\end{document}